\documentclass[letterpaper,10pt,conference]{ieeeconf}

\IEEEoverridecommandlockouts
\usepackage{macros}
\usepackage{xurl}

\title{\LARGE \bf
Min-Max Grassmannian Optimization for Online Subspace Tracking$^*$%
\thanks{$^*$This research is supported by the National Research Foundation of
Singapore under its NRF Fellowship Programme Award No.:
NRF-NRFF18-2026-0010.}}

\author{S. Bharadwaj, B. Mishra,
C. Mostajeran, A. Padoan, J. Coulson
and R. N. Banavar%
\thanks{S. Bharadwaj and R. N. Banavar are with the Centre for Systems
and Control, Indian Institute of Technology Bombay, India.
B.~Mishra is with Microsoft India.
S. Bharadwaj and C. Mostajeran are with the School of Physical
and Mathematical Sciences, Nanyang Technological University, Singapore.
A. Padoan is with the Department of Electrical and Computer Engineering,
University of British Columbia, Vancouver, Canada.
J. Coulson is with the Department of Electrical and Computer Engineering,
University of Wisconsin--Madison, Madison, WI, USA.
Emails:
\nolinkurl{shreyasn001@e.ntu.edu.sg},
\nolinkurl{bamdevm@microsoft.com},
\nolinkurl{cyrussam.mostajeran@ntu.edu.sg},
\nolinkurl{apadoan@ece.ubc.ca},
\nolinkurl{jeremy.coulson@wisc.edu},
\nolinkurl{banavar@iitb.ac.in}}}

\begin{document}

\maketitle
\thispagestyle{empty}
\pagestyle{empty}

\begin{abstract}
We propose GeRoST (Geometrically Robust Subspace Tracking), an
online subspace tracking algorithm that models uncertainty in a
subspace using a Grassmannian ball. We derive an exact
scalar dual for the worst-case subspace problem, establish
conditions for a unique worst-case subspace and a Riemannian
gradient, and characterize the minimum radius needed to cover
a dimensional extension of the target subspace. Each update uses either a spectral direction computed in a reduced subspace or the gradient of the window reconstruction loss. Our numerical experiments show that GeRoST achieves lower mean post-fault prediction error than GREAT in system identification. In video separation, it achieves higher precision and a better precision--recall balance, as measured by the
F$_1$ score, than both GREAT and GRASTA at the reported thresholds, with lower recall and longer runtime.
\end{abstract}

\section{INTRODUCTION}


Online subspace tracking estimates a time-varying low-dimensional subspace from a stream of high-dimensional observations. In video analysis and system identification, these estimates provide models that adapt to changing dynamics and backgrounds without reprocessing all past observations~\cite{balzano_grouse_2010, sasfi2025}. GROUSE~\cite{balzano_grouse_2010} introduced intrinsic updates on the
Grassmannian, while GRASTA~\cite{he2012} combined Grassmannian optimization
with an $\ell_1$-residual model to handle sparse outliers. ReProCS-type
methods~\cite{vaswani2018reprocs,Narayanamurthy_2019,
qiu2011reprocsmissinglinkrecursive,hanguo2014pracreprocs} use recursive
projected compressive sensing to separate a low-rank component from sparse
corruptions. GREAT~\cite{sasfi2025} provides a recursive Grassmannian method
for tracking subspaces associated with linear time-varying systems, together
with analysis under explicit drift, noise and excitation assumptions. Related
online methods on flag manifolds were studied in~\cite{jin2025onlinesubspacelearningflag}.

This paper takes a complementary approach: rather than prescribing a coordinate
structure for the corruption, we place a Grassmannian uncertainty ball around
the dominant subspace of the current data window. Geometric uncertainty sets
of this type have been studied in~\cite{coulson2025,grass_minimum2018}, and
related static min--max problems appear in~\cite{bharadwaj2025}. Combining
this uncertainty model with a sliding window of time-series data gives an online tracker whose
inner maximization problem remains spectral.

The main contributions of the paper are as follows.
We formulate robust online subspace tracking as a
Grassmannian min--max problem, derive its exact scalar dual
and gradient characterization, and establish a minimum
covering-radius result. We evaluate the resulting tracker
against GREAT in system identification and against GREAT
and GRASTA in video separation.

The remainder of the paper is organized as follows. Section~\ref{sec:problem}
introduces the uncertainty model after a brief review of Grassmannian geometry.
Section~\ref{sec:methodology} develops the spectral characterization, the
GeRoST algorithm, and the target-subspace coverage result. Section~\ref{sec:exp}
presents the numerical experiments, followed by concluding remarks. Proofs of
the theoretical results are presented in the Appendix.

\section{PRELIMINARIES}

\subsection{Notation}
The symbols $\Z$ and $\R$ denote integers and reals. We write $\|A\|_F$ for the Frobenius norm and $\sigma_i(A)$ for singular values in decreasing order. For symmetric $A$, $\mu_i(A)$ are its decreasingly ordered eigenvalues, $\operatorname{tr}(A)$ its trace, and $\delta_p(A)=\mu_p(A)-\mu_{p+1}(A)$. Let $V_d(A)$ contain orthonormal eigenvectors associated with the $d$ largest eigenvalues of $A$, and write $\Sv_d(A)=\Span(V_d(A))$.

\subsection{Grassmannian Geometry}

The Grassmannian $\Gr(k,n)$ is the set of $k$-dimensional subspaces in $\R^n$. A subspace $\Su\in\Gr(k,n)$ is represented by a basis
$U\in\R^{n\times k}$, with $U^\top U=I_k$, and the orthogonal projector is given by $P_{\Su}=UU^\top$. A tangent direction at $\Su = \Span(U)$ is represented by a matrix $H \in \R^{n \times k}$, referred to as a horizontal tangent matrix, which satisfies $U^\top H = 0$. For a smooth function $f: \Gr(k,n) \to \R$, its Riemannian gradient is represented at $U$ by 
\[
 \grad \ f(U)=P_{\Su}^{\perp}\nabla\bar f(U),
 \qquad P_{\Su}^{\perp}=I-UU^\top,
\]
where $\bar{f}$ is a smooth Euclidean extension.
For
$\Su\in\Gr(k,n)$ and $\Sv\in\Gr(d,n)$, the chordal distance is~\cite[Table 3]{ye2016}
\begin{equation}\label{eq:chordal}
    d_c(\Su,\Sv) := \Big(|k-d| + \sum_{i=1}^{\min(k,d)} 
    \sin^2\theta_i\Big)^{1/2},
\end{equation}
where $\theta_i$ are the principal angles. If $k=d$, then
$d_c(\Su,\Sv)=\|P_\Su-P_\Sv\|_F/\sqrt2$. For $k\leq d$,
$d_c^2(\Su,\Sv)=\|P_{\Su}^\perp P_{\Sv}\|_F^2$. $\Exp_\Su$ denotes the Grassmannian exponential map at $\Su$. We also write $\Exp_U(H)$ for an orthonormal basis
representing $\Exp_{\Su}(H)$.
\section{PROBLEM STATEMENT}\label{sec:problem}

We model uncertainty in the subspace inferred from a sliding data window by a chordal ball on the Grassmannian. The resulting worst-case loss accounts for deviations between this nominal subspace and the unknown data-generating (or target) subspace. 

Consider a sliding window of the most recent $T$ samples,
\[
 W_t=[w_{t-T+1}\ \cdots\ w_t]\in\R^{n\times T},
\]
where $w_t=\bar w_t+e_t$, $\bar w_t\in\Su_t$, and $e_t\in\R^n$ is an unknown error term. Here $\Su_t\in\Gr(k,n)$ is the data-generating subspace and we denote by
$\Suhat_t\in\Gr(k,n)$ an estimate of $\Su_t$. Throughout,
$1\leq k\leq d\leq n-1$ and $T \geq d$, where $k$ is the target subspace dimension, $d$ is a user-defined parameter representing the dimension retained from the data window, and $T$ is the window length.

 Let $\Swhat_t=\Sv_d(W_tW_t^\top)\in\Gr(d,n)$ be a top-$d$ left singular subspace of the window. Let $\Sy \in \Gr(k,n)$ denote a candidate estimate of the unknown subspace $\Su_t$. We evaluate each candidate by its worst-case projection loss (squared chordal distance) over a ball around $\Swhat_t$, i.e., we study the problem:
\begin{equation}\label{eq:subspace-grls}
 \min_{\Sy\in\Gr(k,n)}F_t(\Sy),
 \qquad
 F_t(\Sy):=\max_{\Sw\in\B_{\rho_t}(\Swhat_t)}f(\Sy,\Sw),
\end{equation}
where $f(\Sy,\Sw):=\norm{P_{\Sy}^{\perp}P_{\Sw}}_F^2$,
and the uncertainty ball is $\B_{\rho_t}(\Swhat_t)
 :=\{\Sw\in\Gr(d,n):d_c(\Sw,\Swhat_t)\leq\rho_t\}$.
The inner maximization selects the subspace within the ball with the largest projection loss, and the outer minimization seeks a $k$-dimensional estimate that minimizes this worst-case loss. In foreground--background separation, moving
objects perturb the low-rank $k$-dimensional background and
$\rho_t$ controls the allowed deviation from the nominal
subspace $\Swhat_t$~\cite{balzano_grouse_2010}.


\begin{remark}[Choice of $d$]
The target subspace has dimension $k$, while the nominal
subspace $\Swhat_t$ is spanned by the $d$ leading (or top) left
singular vectors of the data matrix $W_t$.
Choosing $d>k$ can reduce the minimum covering radius needed for
the uncertainty ball to contain a $d$-dimensional subspace
that includes the target subspace $\Su_t$
(Lemma~\ref{lem:coverage}), but $\Swhat_t \in \Gr(d,n)$ then contains
many possible $k$-dimensional estimates.
In particular, when $\rho_t=0$, every
$\Sy\subseteq\Swhat_t$ with $\dim\Sy=k$ minimizes the
loss, attaining the same value $d-k$. By contrast, GREAT's loss function $\norm{P_{\Sy}^\perp W_t}_F^2$ weights the window directions by their squared singular values~\cite{sasfi2025}.
\end{remark}

\section{ROBUST SUBSPACE TRACKING}\label{sec:methodology}


We begin with the inner maximization. For a fixed $\Sy$, the inner objective in \eqref{eq:subspace-grls} is $d-\tr{P_{\Sy} P_{\Sw}}$, and its constraint is $\textrm{tr}\big(P_{\Swhat_t}P_{\Sw}\big) \geq d-\rho_t^2$. Introducing a multiplier $\lambda \geq 0$ for this constraint, the terms involving $\Sw$ combine as
$\textrm{tr}\bigl((\lambda P_{\Swhat_t}-P_{\Sy})
P_{\Sw}\bigr)$, motivating the definition of $B_t$
in~\eqref{eq:bmatrix}. The inner maximum is attained because its feasible set is compact and its objective is continuous. 


\begin{lemma}[Spectral gap]\label{lem:gap}
For every $\Sy \in \Gr(k,n) $ and $\lambda \geq 0$, the matrix 
\begin{equation}\label{eq:bmatrix}
 B_t(\Sy,\lambda):=\lambda P_{\Swhat_t}-P_{\Sy},
\end{equation}
satisfies the following:
\begin{equation}\label{eq:gap_bound}
 \mu_d(B_t)\geq\lambda-1,\quad
 \mu_{d+1}(B_t)\leq0,\quad
 \delta_d(B_t)\geq\lambda-1.
\end{equation}
Consequently, for $\lambda>2$, the top $d$-dimensional
eigenspace is unique and
\begin{equation}\label{eq:eigenspace_distance}
 d_c\bigl(\Sv_d(B_t),\Swhat_t\bigr)
 \leq{\sqrt{k}}/{(\lambda-1)}
 \leq{\sqrt{k}}/{(\lambda-2)}.
\end{equation}
\end{lemma}

For $\lambda > 2$, Lemma~\ref{lem:gap} gives a unique top-$d$ subspace and an explicit bound on its distance from the nominal subspace. The following theorem connects this eigenspace to the inner maximization through a scalar dual.

\begin{theorem}[Inner maximization and scalar dual]\label{thm:inner}
Let $\Sy\in\Gr(k,n)$ and $\rho_t>0$. Define
\begin{equation}\label{eq:scalar_dual}
 g_t(\Sy,\lambda)
 :=d-\lambda(d-\rho_t^2)
   +\sum_{j=1}^{d}\mu_j\bigl(B_t(\Sy,\lambda)\bigr).
\end{equation}
Then $F_t(\Sy)=\min_{\lambda\geq0}g_t(\Sy,\lambda)$. For each fixed $\Sy$, the scalar function $\lambda\mapsto g_t(\Sy,\lambda)$ attains its minimum over $[0,\infty)$, and every minimizing multiplier lies in $[0,k/\rho_t^2]$.
For every optimal multiplier $\lambda^*$, there is an inner maximizer $\Sw^*$ whose projector maximizes $\tr{P_{\Sw^*}B_t(\Sy,\lambda^*)}$ over all rank-$d$ orthogonal projectors and satisfies
\begin{equation}\label{eq:kkt_recovery}
 \textrm{tr}\big(P_{\Swhat_t}P_{\Sw^*}\big)\geq d-\rho_t^2,\quad
 \lambda^*\bigl(d-\rho_t^2-\textrm{tr}\big(P_{\Swhat_t}P_{\Sw^*}\big)\bigr)=0.
\end{equation}
If $\delta_d(B_t(\Sy,\lambda^*))>0$, this projector is unique and its range is
$\Sv_d(B_t(\Sy,\lambda^*))$. In particular, if some $\bar\lambda>2$ satisfies
\begin{equation}\label{eq:boundary_certificate}
 d_c\bigl(\Sv_d(B_t(\Sy,\bar\lambda)),\Swhat_t\bigr)=\rho_t,
\end{equation}
then the unique inner maximizer is
\begin{equation}\label{eq:closed}
 \Sw_t^*(\Sy)=\Sv_d(B_t(\Sy,\bar\lambda)).
\end{equation}
\end{theorem}

Theorem \ref{thm:inner} expresses the inner maximum $F_t(\Sy)$ in~\eqref{eq:subspace-grls} as a scalar minimization over $\lambda \geq 0$. For fixed $\Sy$ and $\lambda > 2$, we define the residual
\begin{equation}\label{eq:bisect_function}
 h_t(\Sy,\lambda)
 :=d_c\bigl(\Sv_d(B_t(\Sy,\lambda)),\Swhat_t\bigr)-\rho_t.
\end{equation}
For $\rho_t>0$, the map $\lambda \mapsto h_t(\Sy,\lambda)$ is continuous and nonincreasing on $(2,\infty)$,
with $h_t(\Sy,2+\sqrt{k}/\rho_t)\leq0$.
A sign change permits bisection, and an exact root gives
the unique inner maximizer in~\eqref{eq:closed}.
Such a root need not exist. At a dual optimum with
$\mu_d(B_t)=\mu_{d+1}(B_t)$, a leading rank-$d$ projector
must be selected to satisfy~\eqref{eq:kkt_recovery}.

We next characterize how $F_t$ changes with the candidate subspace, including when its inner maximizer is nonunique.

\begin{lemma}[Directional derivative and gradient]\label{lem:gradient}
Fix $\rho_t \geq 0$ and $\Swhat_t \in \Gr(d,n)$. For each $\Sy \in \Gr(k,n)$ let $\mathcal M_t(\Sy)$ denote the set of inner maximizers of~\eqref{eq:subspace-grls} for that $\Sy$. Then $F_t$ is locally Lipschitz and directionally differentiable at every point.
Let $Y$ be an orthonormal basis of $\Sy$. For a horizontal tangent matrix $H$ at $Y$, with $Y^\top H = 0$,
\[
 F_t'(Y;H)=
 \max_{\Sw\in\mathcal M_t(\Sy)}
 \left\langle-2P_{\Sy}^{\perp}P_{\Sw}Y,H\right\rangle_F,
\]
where $F_t'(Y; H)$ denotes the one-sided derivative along $\Exp_{\Sy}(sH)$ at $s=0^+$. If the inner maximizer $\Sw_t^*(\Sy)$ of~\eqref{eq:subspace-grls} is unique, then $F_t$ is differentiable at $\Sy$. Writing $F_t(Y) = F_t(\Span(Y))$, its Riemannian gradient at $Y$ is 
\begin{equation}\label{eq:gradient}
\grad \ F_t(Y)
 =-2P_{\Sy}^{\perp}P_{\Sw_t^*(\Sy)}Y.
\end{equation}
If more than one subspace attains the inner maximum in~\eqref{eq:subspace-grls}, the Clarke subdifferential of $F_t$ at $\Sy$, represented at $Y$, is the convex hull
of $-2P_{\Sy}^{\perp}P_{\Sw}Y$ over
$\Sw\in\mathcal M_t(\Sy)$.
\end{lemma}


\subsection{Algorithm}\label{sec:algorithm}

We present the Geometrically Robust Subspace Tracking (GeRoST) algorithm (Algorithm~\ref{alg:gerost}) to track the unknown subspace $\Su_t$ using the uncertainty model in~\eqref{eq:subspace-grls}.
For each new sample, GeRoST performs $K$ updates, each using either the spectral direction in~\eqref{eq:gradient}, with a subspace obtained by searching for a zero
of~\eqref{eq:bisect_function} over $\lambda>2$, or the gradient of the window loss $D_t(\Sy):=\norm{P_{\Sy}^\perp W_t}_F^2$.

\textit{Window initialization and update}: During warm-up, the data window contains $q_t \leq T$ samples. The estimate $\Suhat_t$ is held fixed at $\Suhat_{t-1}$ while $q_t < d$. Once $q_t \geq d$, updates begin using a rank-$d$ SVD of the data window. Once the window is full, the implementation maintains a truncated eigendecomposition of $C_t = W_t W_t^\top$ using the rank-$2$ covariance update in GREAT~\cite[Rem. 5]{sasfi2025}: $C_t = C_{t-1} - w_{t-T}w_{t-T}^\top + w_t w_t^\top$.
The truncated eigendecomposition approximates the leading $d$ eigenvectors of $C_t$. The corresponding orthonormal vectors form $\widetilde W_t$ and span the computed center
$\widetilde{\Sw}_t$. The recursion for $C_t$ is exact, but the truncation can make $\widetilde{\Sw}_t$ differ from $\Swhat_t$. Algorithm~\ref{alg:gerost} starts with a full data window.

\textit{Reduced scalar search}: 
Let $Y_i$ be the orthonormal basis representing the
$i$th inner iterate $\Sy_i=\Span(Y_i)$ in
Algorithm~\ref{alg:gerost}.
The routine \texttt{ReducedInner} forms an orthonormal
basis $Z_i$ of $\Span([Y_i\ \widetilde W_t])$.
With $\Sy=\Sy_i$ and $\Swhat_t$ replaced by the computed
center $\widetilde{\Sw}_t$, the matrix
in~\eqref{eq:bmatrix} vanishes on $\Span(Z_i)^\perp$. By~\eqref{eq:gap_bound}, its leading $d$
eigenvalues are positive for $\lambda>2$, so the corresponding
eigenspace lies entirely in $\Span(Z_i)$. The routine therefore
computes this eigenspace by restricting the matrix to the
basis $Z_i$.
\begin{algorithm}[t]
\caption{GeRoST}
\label{alg:gerost}
\begin{algorithmic}[1]
\item \Input{initial estimate $\Suhat_{t_0}$, full window $W_{t_0}$, samples $w_t$, dimensions
$k,d$, window length $T\geq d$, update count $K$, radius parameters $\bar\xi, c_{\rho}$, proxy cap $\xi_{\max} \in (0,1)$, clipping margin $\varepsilon_\rho \in (0,\sqrt{d}/2)$,  
and an optional fixed stepsize}
\For{$t=t_0+1,t_0+2,\ldots$}
    \State Compute $\widehat\xi_t$ by \eqref{eq:xi_proxy}, radius $\rho_t$ by \eqref{eq:rho_proxy}
    \State Update $W_t$ and
    $\widetilde{\Sw}_t$ as described in Sec.~\ref{sec:algorithm}
    \State Set $Y_0$ to an orthonormal basis of $\Suhat_{t-1}$
    \For{$i=0,\ldots,K-1$}
        \State $\Sy_i \gets \Span(Y_i)$
        \State $(\Sw_i,\lambda_i^{\rm alg})\gets
        \texttt{ReducedInner}(Y_i,\widetilde{\Sw}_t,\rho_t)$
        \If{$\lambda_i^{\rm alg}>2$}
            \State $G_i\gets-2P_{\Sy_i}^{\perp}P_{\Sw_i}Y_i$
        \Else
            \State $G_i\gets-2P_{\Sy_i}^{\perp}W_tW_t^\top Y_i$
        \EndIf
        \State Choose $\alpha_i$ by~\eqref{eq:stepsize}, or use the fixed value
        \State $Y_{i+1}\gets\Exp_{Y_i}(-\alpha_iG_i)$,
        $\Sy_{i+1}=\Span(Y_{i+1})$
    \EndFor
    \State $\Suhat_t\gets\Sy_K$
\EndFor
\item \Output{estimates $\{\Suhat_t\}_{t>t_0}$}
\end{algorithmic}
\end{algorithm}
Denote the resulting subspace by $\Sw_i(\lambda)$ and the corresponding residual in \eqref{eq:bisect_function} by $\tilde{h}_i(\lambda)$. The initial search interval is
\[
 \lambda_{\rm lo}=2+10^{-6},
 \qquad
 \lambda_{\rm hi}=2+{\sqrt{k}}/{\max\{\rho_t,10^{-10}\}}.
\]
If $\tilde{h}_i(\lambda_{\rm lo})\leq0$, the routine returns
the flag $\lambda_i^{\rm alg}=2$ to select the window-loss update. The subspace $\Sw_i$ is then unused. Otherwise, if $\tilde{h}_i(\lambda_{\rm hi}) \geq 0$, it returns $\lambda_i^{\rm alg} = \lambda_{\rm hi}$. In the remaining case, the endpoint residuals have opposite signs and the routine applies bisection. It stops when the interval width is at most $10^{-6}$ or the iteration limit \texttt{iter} is reached, whichever occurs first, and returns
the midpoint of the final interval. Whenever
$\lambda_i^{\rm alg}>2$, the returned subspace is
$\Sw_i=\Sw_i(\lambda_i^{\rm alg})$. In our numerical experiments, we obtain  satisfactory performance using $\texttt{iter}=32$.


An exact root with $\lambda > 2$ satisfies \eqref{eq:boundary_certificate} for the computed center $\widetilde{\Sw}_t$ and gives the unique maximizer \eqref{eq:closed}. Consequently, \eqref{eq:gradient} gives the robust gradient. The routine approximates this root using finitely many bisection steps. Endpoint returns serve as numerical safeguards and do not
by themselves ensure \eqref{eq:boundary_certificate}. The flag
$\lambda_i^{\rm alg}=2$ selects the window-loss update and does not imply dual optimality.

\textit{Stepsize}: For $\lambda_i^{\rm alg} > 2$, define
\begin{equation}\label{eq:numerical_constants}
\begin{split}
 \Delta_i&=\max\{\lambda_i^{\rm alg}-2,10^{-2}\},\\
 L_i&=4(1+\Delta_i^{-1})+4\sqrt d\,\Delta_i^{-2}.\\
\end{split}
\end{equation}
For the window-loss update, set
$L_i=4\widetilde\sigma_{1,t}^2$, where
$\widetilde\sigma_{1,t}$ is the leading window singular value estimate obtained from the truncated eigendecomposition of $W_tW_t^\top$ described above. Unless a fixed stepsize is supplied, we use
\begin{equation}
   \label{eq:stepsize}
   \alpha_i = {1}/{\max\{L_i, 10^{-8}\}}.
\end{equation}
We use this rule for adaptive scaling and observe good performance in the reported experiments. Establishing conditions for descent or contraction is left for future work.

\begin{remark}[Computational cost]\label{rem:cost}
Let $b_{\rm dual}$ be the number of spectral evaluations in one scalar search and let $r = \dim \Span(Z_i)\leq\min\{n,d+k\}$. Forming the reduced basis and performing the search costs $O(nr^2+b_{\rm dual}r^3)$ per inner update. The window-loss direction uses $W_t(W_t^\top Y_i)$ at cost $O(nTk)$. These costs exclude maintenance of computed center
and exponential-map updates. With low-rank covariance maintenance, the complete algorithm has $\mathcal{O}(n)$ cost per sample when $T,d,k,K$ and the search budget are fixed.
\end{remark}

\subsection{Covering the target subspace}\label{sec:coverage}

When $d>k$, $\Su_t$ and the elements
of the uncertainty ball have different dimensions.
Coverage therefore means that the ball contains some
$\Sv\in\Gr(d,n)$, with $\Su_t\subseteq\Sv$, and we call $\Sv$ a $d$-dimensional extension of
$\Su_t$.

\begin{lemma}[Minimum covering radius]\label{lem:coverage}
Let $U_t$ be an orthonormal basis of $\Su_t\in\Gr(k,n)$. Then
\begin{equation}\label{eq:coverage}
 r_t^{\rm cov}
 :=\min_{\substack{\Sv\in\Gr(d,n)\\\Su_t\subseteq\Sv}}
 d_c(\Sv,\Swhat_t)
 =\norm{P_{\Swhat_t}^{\perp}U_t}_F.
\end{equation}
Thus $\B_{\rho_t}(\Swhat_t)$ contains a $d$-dimensional extension of
$\Su_t$ if and only if $\rho_t\geq r_t^{\rm cov}$. If $\Swhat_t$ is an
exact top-$d$ left singular subspace of $W_t$, $\sigma_k(W_t)>0$, and
$\norm{P_{\Su_t}^{\perp}W_t}_F\leq\eta_t$, then
\begin{equation}\label{eq:coverage_bound}
 r_t^{\rm cov}\leq{\eta_t}/{\sigma_k(W_t)}.
\end{equation}
When $d=k$, containment reduces to coverage of $\Su_t$ itself.
\end{lemma}

The minimum covering radius and its bound in
Lemma~\ref{lem:coverage} depend on the unknown target
subspace. We therefore select
the radius empirically using the observable pre-update proxy
\begin{equation}\label{eq:xi_proxy}
\widehat\xi_t
={\|(I-P_{\Suhat_{t-1}})w_t\|_2}/
{\max\{\|\widehat U_{t-1}^\top w_t\|_2,\varepsilon_w\}},
\end{equation}
where $\widehat U_{t-1}$ is an orthonormal basis of
$\Suhat_{t-1}$, and $\varepsilon_w>0$ prevents division by zero. This proxy compares the sample's residual norm with the norm
of its projection onto the previous subspace estimate.
Given a cap $\bar\xi>0$, a scale $c_\rho>0$, and safeguards
$0<\xi_{\max}<1$, $0<\varepsilon_\rho<\sqrt d/2$, set
$
\widetilde\xi_t
=\min\{\widehat\xi_t,\bar\xi,\xi_{\max}\}
$
and 
\begin{equation}\label{eq:rho_proxy}
\rho_t
=\min\Big\{\sqrt d-\varepsilon_\rho,\,
\max\Big\{\varepsilon_\rho,\,
c_\rho\widetilde\xi_t/(1-\widetilde\xi_t)
\Big\}\Big\}.
\end{equation}
In the experiments, we use $\xi_{\max}=0.99$ and
$\varepsilon_w=\varepsilon_\rho=10^{-6}$.
The radius is nondecreasing in $\widehat\xi_t$, but this
empirical rule does not guarantee the coverage condition
in Lemma~\ref{lem:coverage}.

\section{NUMERICAL EXPERIMENTS}\label{sec:exp}

We compare GeRoST with GREAT~\cite{sasfi2025} in system identification and
video separation, and with GRASTA~\cite{he2012} in video foreground--background separation. The implementation
uses \textit{Manopt}~\cite{manopt} and the experiment code is available at
\url{https://github.com/shreyasnb/gerost}.

\subsection{Case 1: Online system identification}

We adapt the setup in~\cite[Sec.~5]{sasfi2025} to a mass--spring--damper system
with sinusoidally time-varying spring stiffness. The ambient dimension is
$n=20$, the true subspace dimension is $k=13$, and GeRoST uses $d=16$. The
window length is $T=100$, with $K=3$ inner updates per time step.

The estimated subspace is used to produce a prediction denoted $\hat{y}_t$ of the true system response, $y_t$, to a new input. Following~\cite{sasfi2025}, the relative prediction error at time $t$ is
$
e_t^{\mathrm{pred}}
:=
\left(
\frac{
\sum_{i=t+1}^{t+T_{\mathrm{fut}}}
\|\hat y_i-y_i\|_2^2
}{
\sum_{i=t+1}^{t+T_{\mathrm{fut}}}
\|y_i\|_2^2
}
\right)^{1/2},
$
where \(y_i\) and \(\widehat y_i\) are the true and predicted outputs,
respectively, and \(T_{\rm fut}\) is the prediction horizon. At
$t=80$, a Gaussian sensor spike is injected with standard deviation ten times
the nominal noise level, corresponding approximately to an SNR (signal-to-noise ratio) decrease from
$14\,\mathrm{dB}$ to $-6\,\mathrm{dB}$. Figure~\ref{fig:sysid} shows the pointwise mean and $\pm1$ sample standard deviation over $50$ trials. Pooling prediction errors over these trials and all post-fault indices \(t=81,\ldots,T_{\rm test}-T_{\rm fut}\), where $T_{\rm test}$ is the number of test samples, GeRoST attains a lower mean post-fault prediction error than GREAT
($1.680$ versus $5.480$) and a lower standard deviation
($1.324$ versus $11.877$).

\begin{figure}[t]
 \centering
 \includegraphics[width=0.8\linewidth]{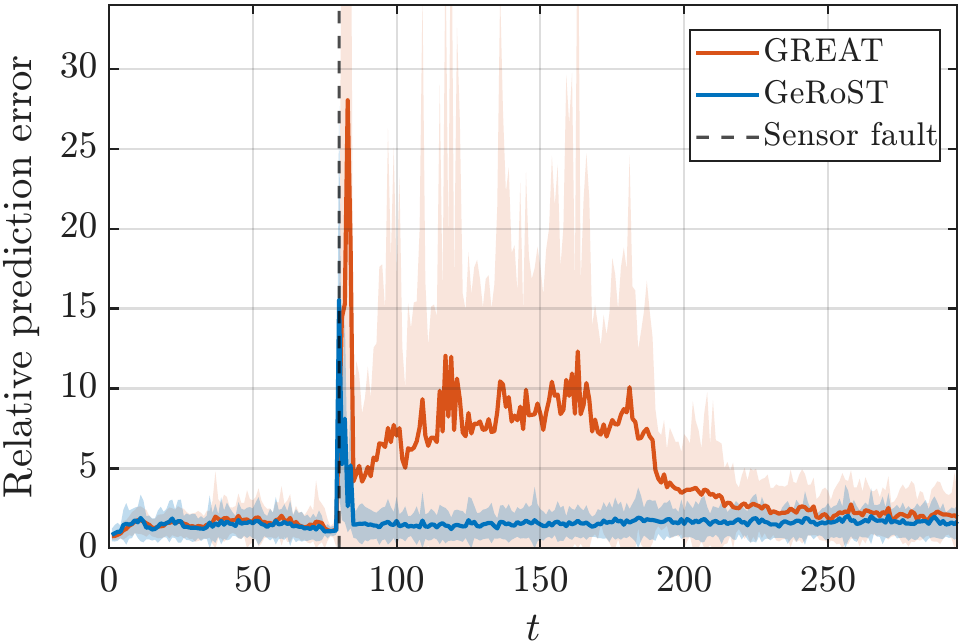}
 \caption{Relative prediction error for GeRoST and GREAT over \(50\) Monte Carlo trajectories. Bold curves show the sample mean, and shaded regions indicate \(\pm\) one sample standard deviation.}
 \label{fig:sysid}
\end{figure}

\subsection{Case 2: Online foreground--background separation}

We use the \texttt{baseline/highway} sequence of CDnet
2014~\cite{cdnet2014}, consisting of \(1700\) grayscale frames of size
\(n = 240\times320\) and target rank $k=5$. Frames are scaled to \([0,1]\). For each method, the
pixelwise foreground score is $s_{p,t}=|w_{p,t}-\widehat b_{p,t}|$, where \(w_{p,t}\) is the observed intensity at pixel \(p\) and time \(t\),
and \(\widehat b_{p,t}\) is the corresponding reconstructed background
intensity. All methods receive the same initialization from frames
1--50. GREAT and GeRoST use a window of length \(30\) and \(K=5\)
updates per frame; GeRoST uses \(d=7\). 

Evaluation uses the Highway temporal Region of Interest (ROI), frames 470--1700, together
with the supplied spatial ROI. Ground-truth label 255 is treated as foreground,
labels 0 and 50 as negative, and labels 85 and 170 are excluded. Classification metrics are computed from predictions
pooled over all valid pixels and frames. For each method, a single
threshold is calibrated without ground truth from frames outside the temporal ROI as
\[
 \tau=\max\{0.02,\operatorname{median}(s)
       +6(1.4826)\operatorname{MAD}(s)\},
\]
where $s$ denotes foreground scores in the spatial ROI. $\operatorname{MAD}$ is the median absolute
deviation, and $1.4826$ is the Gaussian consistency factor. The area under the receiver operating characteristic (ROC) curve, AUC, is computed from a 1024-bin histogram sweep of the raw
foreground scores over the same valid evaluation pixels. Runtime is the
complete tracker call without the diagnostic overhead, and the frame rate in frames per second (FPS) is the number
of loaded frames divided by this runtime.

Recall measures the fraction of true foreground pixels
detected, precision the fraction of predicted foreground
pixels that are correct, and $F_1$ their harmonic mean.
GRASTA achieves highest AUC and recall. GeRoST achieves higher $F_1$ and precision and a lower percentage of wrong classifications (PWC) than GREAT (see Table~\ref{tab:case2}). GREAT and GeRoST have similar AUCs but different precision--recall tradeoffs (see ROC curves in Fig. \ref{fig:roc}). 

\begin{table}[t]
\caption{CDnet \texttt{baseline/highway} results on the temporal and spatial
evaluation ROI.}
\label{tab:case2}
\centering
\begin{tabular}{l|cccccc}
\hline
 & AUC & F$_1$ & Precision & Recall & PWC & FPS \\
\hline
GRASTA & \textbf{0.9716} & 0.554 & 0.392 & \textbf{0.942} & 9.002 & \textbf{25.42} \\
GREAT  & 0.9462 & 0.535 & 0.388 & 0.861 & 8.872 & 5.49 \\
GeRoST & 0.9469 & \textbf{0.567} & \textbf{0.427} & 0.842 & \textbf{7.631} & 4.83 \\
\hline
\end{tabular}
\end{table}
\begin{figure}[t]
 \centering
 \includegraphics[width=0.8\linewidth]{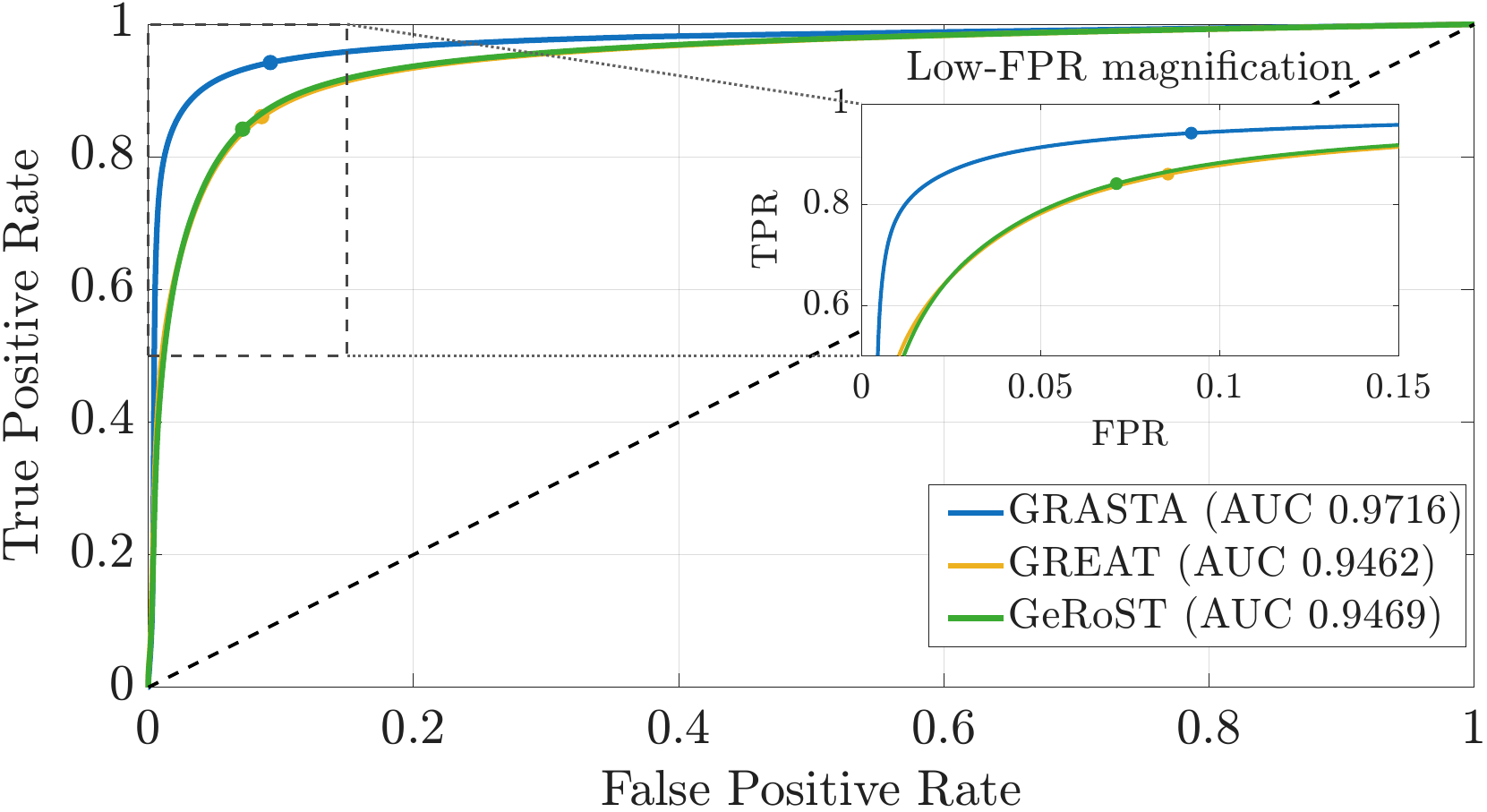}
 \caption{Pixel-wise ROC curves on the valid CDnet evaluation pixels.}
 \label{fig:roc}
\end{figure}

\begin{figure}[t]
 \centering
 \includegraphics[width=\linewidth]{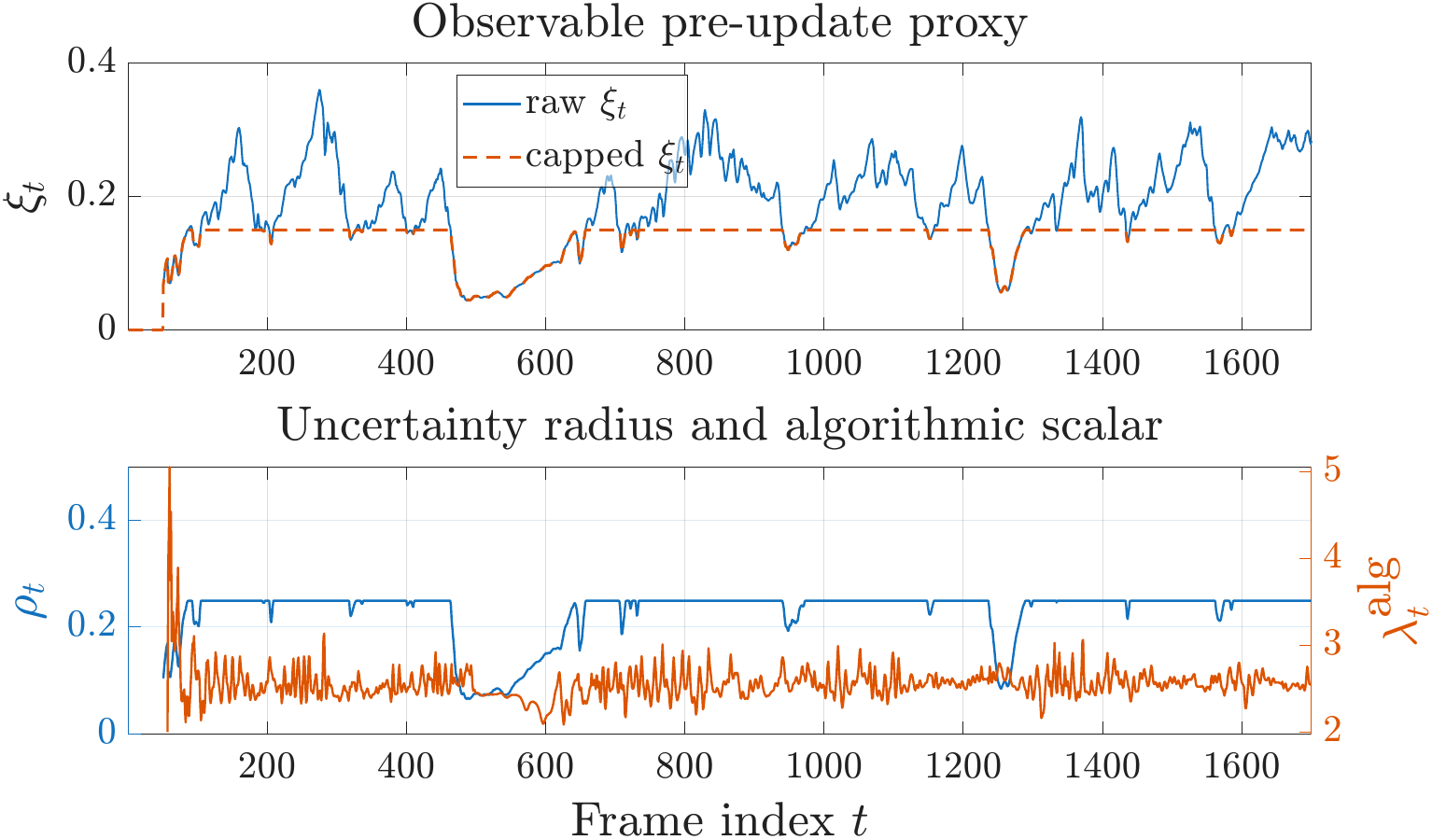}
 \caption{GeRoST diagnostics. \textbf{Top:} pre-update proxy.
\textbf{Bottom:} uncertainty radius (left axis) and $\lambda_t^{\mathrm{alg}} = \lambda_{K-1}^{\rm alg}$ (right axis) returned by the final update.}
 \label{fig:rho_lambda}
\end{figure}

Figure~\ref{fig:rho_lambda} shows the proxy, radius, and the update-selection multiplier. Figure~\ref{fig:fg-bg} gives a qualitative comparison at frame $1000$, which
lies in the evaluation ROI. The foreground panels use a common display range based on the pooled 99.5th
percentile, so their intensities are directly comparable. In this snapshot,
GRASTA produces a cleaner background, while GREAT and GeRoST retain visible
vehicle ghosting. Future work will investigate the causes of the difference.

\begin{figure}[h]
 \centering
 \includegraphics[width=0.85\linewidth]{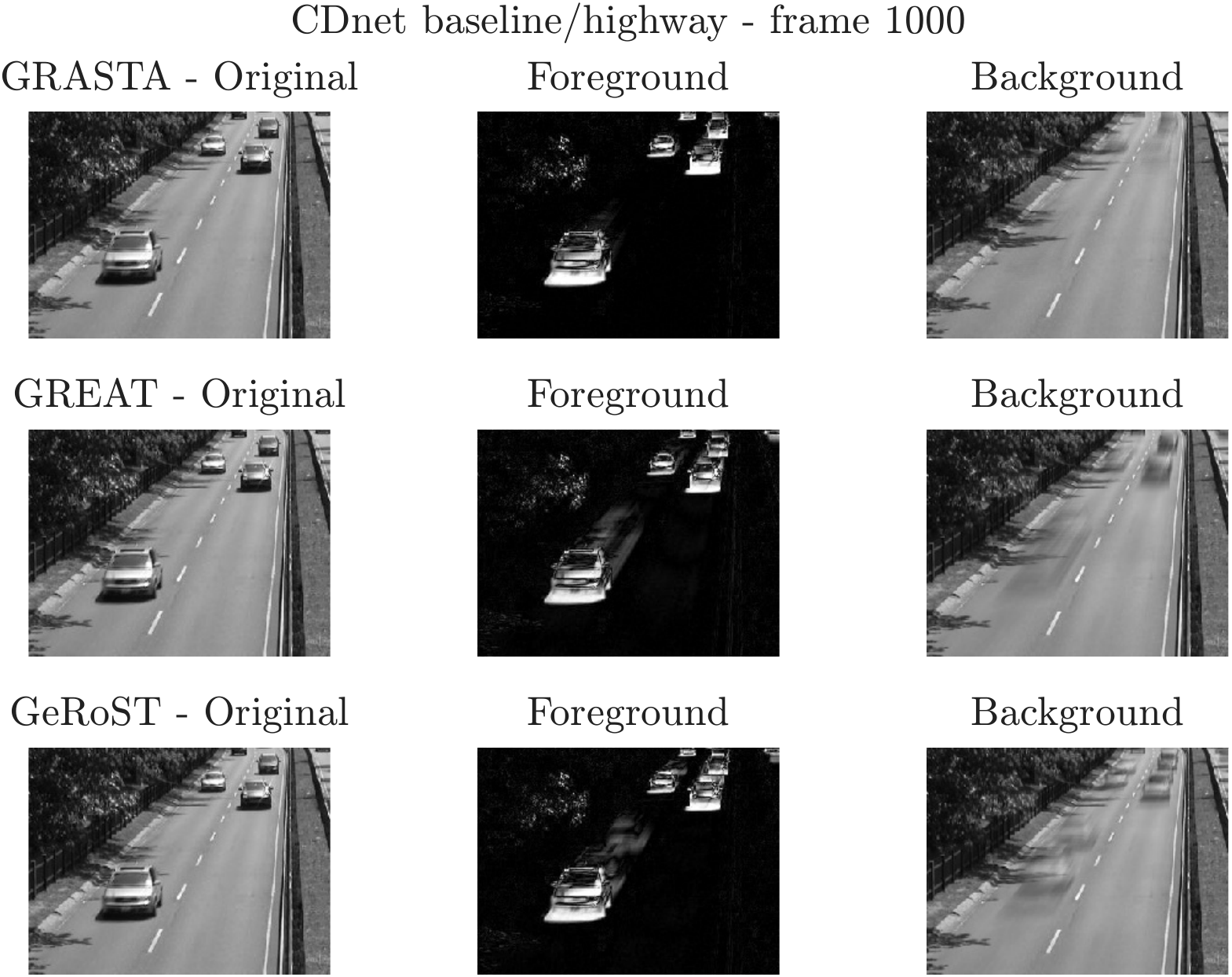}
 \caption{Foreground--background separation at frame $1000$. Rows: GRASTA,
 GREAT and GeRoST. Columns: original frame, foreground score and background
 estimate.}
 \label{fig:fg-bg}
\end{figure}

\section{CONCLUSION}

We introduced GeRoST, a subspace tracker based on a Grassmannian uncertainty
model. The inner problem admits an exact scalar dual, and a positive spectral gap at an optimal multiplier gives a unique spectral
maximizer. Experiments show lower post-fault prediction error when predicting outputs of an unknown dynamical system
relative to GREAT, and a different precision–recall balance on a foreground--background separation task on CDnet Highway. Future work will analyze convergence of GeRoST when switching between the spectral and window-loss updates, and evaluate its performance on additional datasets.

\appendix

\emph{Proof of Lemma~\ref{lem:gap}.}
Let $B=B_t(\Sy,\lambda) = A_0 + E$ where $A_0=\lambda P_{\Swhat_t}$ and $E=-P_{\Sy}$. For $\lambda \geq 0$, Weyl's
inequalities~\cite[Thm.~III.2.1]{bhatia1997matrix} give
\[
\begin{aligned}
 \mu_d(B)
 &\geq \mu_d(A_0)+\mu_n(E)=\lambda-1,\\
 \mu_{d+1}(B)
 &\leq \mu_{d+1}(A_0)+\mu_1(E)=0.
\end{aligned}
\]
Hence $\delta_d(B)\geq\lambda-1$, proving
\eqref{eq:gap_bound}. For $\lambda > 2$, this gap is positive, so the leading
$d$-dimensional subspace is unique. Its eigenvalues
are separated from the zero eigenvalues of $A_0$ on
$\Swhat_t^\perp$ by at least $\lambda-1$.
Applying \cite[Thm.~1]{yu2015} with
$\Sigma=B$ and $\widehat\Sigma=A_0$ gives
%
$
 d_c(\Sv_d(B_t),\Swhat_t)
 \leq\frac{\norm{B-A_0}_F}{\lambda-1} = \frac{\sqrt{k}}{\lambda-1}
 \leq\frac{\sqrt{k}}{\lambda-2},
$
using \eqref{eq:chordal} and $B-A_0 = -P_{\Sy}$.
\hfill$\square$

\emph{Proof of Theorem~\ref{thm:inner}.}
Fix $\Sy$, and set $Q=P_{\Swhat_t}$, $R=P_{\Sy}$. Let
$\mathfrak F_d:=\{X \in \R^{n \times n} : X=X^\top, 0\preceq X\preceq I,\ \tr{X}=d\}$.
Replacing the rank-$d$ projector in the inner maximization of~\eqref{eq:subspace-grls} by $X\in\mathfrak F_d$ gives the relaxation 
\begin{equation}\label{eq:fantope_inner}
 F_t(\Sy) \leq \max_{X\in\mathfrak F_d}
 \{d-\tr{RX}:\tr{QX}\geq d-\rho_t^2\}.
\end{equation}
We first show that the relaxed problem has a rank-$d$ projector optimizer, so equality holds in~\eqref{eq:fantope_inner}.

An extreme point $X$ of the feasible set in~\eqref{eq:fantope_inner} cannot have exactly one eigenvalue in $(0,1)$ because $\tr X=d$ is integral. If it has $r\geq2$ eigenvalues in $(0,1)$, let $E$ contain the corresponding orthonormal eigenvectors. Among the perturbations $H=ESE^\top$, $S=S^\top$, the $r(r+1)/2>2$ degrees of freedom allow a nonzero $H$ with $\tr H=\tr{QH}=0$. Let $\delta > 0$ be the minimum distance of the eigenvalues of $X$ in $(0,1)$ from $\{0,1\}$. Since $H=ESE^\top$ acts only on their eigenspace, choosing $0<\epsilon\|S\|_2<\delta$ ensures $0\preceq X\pm\epsilon H\preceq I$.
The remaining eigenvalues are unchanged, and
$\tr{H}=\tr{QH}=0$ preserves both trace constraints. Thus $X\pm\epsilon H$ are distinct feasible points with midpoint $X$, contradicting extremality. Hence every extreme feasible point is a rank-$d$ projector, and compactness ensures a projector optimizer of the linear objective in \eqref{eq:fantope_inner}.

For $\rho_t>0$,
the matrix $X_\epsilon=(1-\epsilon)Q+\epsilon(d/n)I$ satisfies $0\prec X_\epsilon \prec I$, $\tr{X_{\epsilon}} = d$ and $\tr{QX_{\epsilon}} > d-\rho_t^2$ for sufficiently small
$\epsilon>0$. Since the relaxed problem in~\eqref{eq:fantope_inner} has a finite optimal value, Slater's theorem~\cite[Sec.~5.2.3]{boyd2004} gives strong duality for this problem and attainment of its Lagrange dual optimum. Expressing $B(\lambda) = \lambda Q - R$, diagonalization shows $
    \max_{X \in \mathfrak F_d} \tr{B(\lambda)X} = \sum_{j=1}^d \mu_j(B(\lambda))$,
i.e., in an eigenbasis of $B(\lambda)$, the diagonal entries of $X$ lie in $[0,1]$ and sum to $d$. The upper bound is attained by a leading rank-$d$ projector. Substitution into the Lagrange dual gives
$F_t(\Sy)=\min_{\lambda\geq0}g_t(\Sy,\lambda)$.
Evaluating $\max_{X\in\mathfrak F_d}\tr{B(\lambda)X}$
at $X=Q$ in the expression for $g_t(\Sy,\lambda)$ gives
$g_t(\Sy,\lambda)\geq d-k+\lambda\rho_t^2$.
Since $g_t(\Sy,0)\leq d$, every minimizer of
$\lambda\mapsto g_t(\Sy,\lambda)$ over $\lambda\geq0$
lies in $[0,k/\rho_t^2]$.

For any dual optimum $\lambda^*$, a projector optimizer
$P^*$ of \eqref{eq:fantope_inner} also maximizes
$\tr{B(\lambda^*)X}$ over $X\in\mathfrak F_d$ and
satisfies \eqref{eq:kkt_recovery}, by strong duality
and complementary slackness~\cite[Sec.~5.5.3]{boyd2004}.
A strict spectral gap yields a unique maximizer.

For $\lambda>2$, let $P(\lambda)$ denote the unique leading rank-$d$ projector of $B(\lambda)$. Lemma~\ref{lem:gap} and the Davis--Kahan theorem
\cite[Thm.~1]{yu2015} imply its continuity. For $2<\lambda_1<\lambda_2$, write
$P_j=P(\lambda_j)$ and $B_j=B(\lambda_j)$.
Optimality gives
\[
 \tr{B_1P_1}\geq\tr{B_1P_2},
 \qquad
 \tr{B_2P_2}\geq\tr{B_2P_1}.
\]
Adding these inequalities gives
\[
 (\lambda_2-\lambda_1)
\operatorname{tr}\bigl(Q(P(\lambda_2)-P(\lambda_1))\bigr)
 \geq0.
\]
Since $d_c^2(\Sv_d(B(\lambda)),\Swhat_t) =d-\operatorname{tr}(QP(\lambda))$, the residual~\eqref{eq:bisect_function} is continuous and nonincreasing. Moreover,
\eqref{eq:eigenspace_distance} gives $h_t(\Sy,2+\sqrt{k}/\rho_t)\leq0$.
The same inequalities give
\[
 \tr{QP_1}
 \leq
 \frac{g_t(\Sy,\lambda_2)-g_t(\Sy,\lambda_1)}
      {\lambda_2-\lambda_1}
 -(\rho_t^2-d)
 \leq
 \tr{QP_2}.
\]
Continuity of $P(\lambda)$ therefore gives
\[
 {\partial g_t(\Sy,\lambda)}/{\partial\lambda} 
 =\rho_t^2-d+\tr{QP(\lambda)},
 \qquad \lambda>2.
\]
Suppose $\bar\lambda>2$ satisfies
\eqref{eq:boundary_certificate}.
Then this derivative vanishes at $\bar\lambda$. Since $g_t(\Sy,\cdot)$ is convex as a maximum of affine functions, $\bar\lambda$ is dual optimal. Its leading projector satisfies~\eqref{eq:kkt_recovery}, and the strict gap gives the unique maximizer~\eqref{eq:closed}.
\hfill$\square$

\emph{Proof of Lemma~\ref{lem:gradient}.}
For fixed $t$, the set $\B_{\rho_t}(\Swhat_t)$ is compact and independent of $\Sy$. The function $f(\Sy,\Sw)$ is smooth in $\Sy$, with partial Riemannian gradient $\grad_{\Sy} \ f(Y,\Sw) = -2P_{\Sy}^\perp P_{\Sw} Y$. This gradient is jointly continuous in $Y$ and $\Sw$, so the partial losses have a uniform local Lipschitz bound. Their maximum $F_t$ is therefore locally Lipschitz.

Applying the smooth maximum rule ~\cite[Sec.~2.8]{clarke1990} in local coordinates gives the stated directional derivative as the maximum of the partial derivatives over $\mathcal M_t(\Sy)$. If the maximizing subspace is unique, this rule also gives differentiability and~\eqref{eq:gradient}. More generally, the same rule gives the Clarke subdifferential as the convex hull of these partial gradients over $\mathcal M_t(\Sy)$, represented horizontally at $Y$.
\hfill$\square$

\emph{Proof of Lemma~\ref{lem:coverage}.}
Every $\Sv \in \Gr(d,n)$ containing $\Su_t$ admits an orthonormal basis $[U_t\ C]$, where $C \in \R^{n \times (d-k)}$, $C^\top C = I$ and $U_t^\top C = 0$. By \eqref{eq:chordal},
\[
 d_c^2(\Sv,\Swhat_t)
 =\norm{P_{\Swhat_t}^{\perp}U_t}_F^2
  +\norm{P_{\Swhat_t}^{\perp}C}_F^2.
\]
Since $\dim(\Swhat_t\cap\Su_t^\perp)\geq d-k$, the columns of $C$ can be
chosen in this intersection. The second term above vanishes, proving~\eqref{eq:coverage} and the stated equivalence with $\rho_t \geq r_t^{\rm cov}$. For $d=k$, $C$ has no columns and the only extension is $\Su_t$ itself.

To prove~\eqref{eq:coverage_bound}, choose compatible top-$k$ left singular vectors $V_k$, right singular vectors $Z_k$, and diagonal matrix $\Sigma_k$ of singular values of $W_t$ such that $W_t Z_k = V_k \Sigma_k$ with $\Span(V_k) \subseteq \Swhat_t$.
Such a choice is possible because $\Swhat_t$ is a top $d$-dimensional singular subspace and $k \leq d$, including when singular values are repeated. Since $\sigma_k(W_t) > 0$, $\Sigma_k$ is invertible, and 
\begin{equation*}
        r_t^{\rm cov}
 \leq \norm{P_{\Su_t}^{\perp}V_k}_F
 =\norm{P_{\Su_t}^{\perp}W_tZ_k\Sigma_k^{-1}}_F\leq{\eta_t}/{\sigma_k(W_t)},
\end{equation*}
where the first inequality follows from the inclusion $\Span(V_k)\subseteq\Swhat_t$
and~\eqref{eq:chordal}.
\hfill$\square$

\bibliographystyle{IEEEtran}
\bibliography{ref_short}

\end{document}